\RequirePackage{fix-cm}
\documentclass[11pt,reqno]{amsproc}
\usepackage{fixltx2e}
    \numberwithin{equation}{section}
\usepackage{amssymb}
\usepackage[citation-order,nobysame]{amsrefs}
\usepackage{xyzbib}
\usepackage{hyperref}

\newcommand\rmi{\mathrm{i}}

\newcommand\NW{\mathrm{NW}}
\newcommand\NE{\mathrm{NE}}
\newcommand\SE{\mathrm{SE}}
\newcommand\SW{\mathrm{SW}}

\renewcommand\geq{\geqslant}
\renewcommand\phi\varphi

\newcommand\arctic{\mathcal{A}}
\newcommand\FO{\mathsf{F}}
\newcommand\DI{\mathsf{D}}

\newtheorem{proposition}{Proposition}
\newtheorem{corollary}{Corollary}

\newcommand\preprint[1]{\begin{flushright} \large\texttt{#1} \end{flushright}\bigskip}

\begin{document}
\preprint{DFF 452/02/2010}

\title[Algebraic arctic curves]{Algebraic arctic curves in the domain-wall six-vertex model}

\author{F. Colomo}
\address{INFN, Sezione di Firenze\\
Via G. Sansone 1, 50019 Sesto Fiorentino (FI), Italy}
\email{colomo@fi.infn.it}

\author{V. Noferini}
\address{Department of Mathematics, University of Pisa\\
Largo Bruno Pontecorvo 5, 56127 Pisa, Italy}
\email{noferini@mail.dm.unipi.it}

\author{A. G. Pronko}
\address{Department of Physics, University of Wuppertal,
42097 Wuppertal, Germany
(\textit{On leave of absence from}:
Saint Petersburg Department of V.~A.~Steklov Mathematical Institute,
Russian Academy of Sciences,
Fontanka 27, 191023 Saint Petersburg, Russia)}
\email{agp@pdmi.ras.ru}

\begin{abstract}

The arctic curve, i.e. the spatial curve separating ordered (or `frozen')
and disordered (or `temperate) regions, of the six-vertex model with
domain wall boundary conditions is discussed for
the root-of-unity vertex weights. In these cases the curve is described by
algebraic equations which can be worked out explicitly from the parametric
solution for this curve. Some interesting examples are discussed in
detail. The upper bound on the maximal degree of the equation in a
generic root-of-unity case is obtained.
\end{abstract}
\maketitle

\section{Introduction}

It is commonly known that in strongly correlated systems boundary conditions may
lead to emergence of spatial separation of phases, e.g., order and disorder.
Famous examples of such systems are domino tiling of large Aztec diamonds
\cite{CEP-96,JPS-98}
and lozenge tilings of large hexagon \cite{CLP-98,BGR-10}.
Another example, more elaborated and with important
combinatorial applications as particular cases, is represented by the six-vertex model
with domain wall boundary conditions
\cite{K-82,I-87,ICK-92,E-99,KZj-00,Zj-00,Zj-02,SZ-04,AR-05,R-10}.

Recently, for this model a progress was achieved in finding an explicit
form of the curve separating the disordered and ferroelectrically ordered
phases (called Arctic curve, by analogy with Arctic circle of domino
tilings). In \cite{CP-08} a conjectural expression for the Arctic curve was
derived for various  particular cases of Boltzmann weights, with one of
them providing the limit shape of large alternating-sign matrices. In
\cite{CP-09} this result for the Arctic curve was extended to the case of generic
Boltzmann weights of the six-vertex model corresponding to its disordered
regime. The case of anti-ferroelectric regime was considered in \cite{CPZj-10}.

In this paper we further discuss the Arctic curve of
the domain-wall six-vertex model, focussing on the special cases of Boltzmann
weights in the disordered regime, where the Arctic curve, generally
transcendental, becomes algebraic. The six-vertex model in these cases
is sometimes referred to as the six-vertex model at `roots of unity'
(the term comes mostly from the quantum-group context). Some of these cases
were already considered in \cite{CP-09}. Here we explain in detail
how the algebraic equation can be derived from
the general parametric solution for the curve, when the Boltzmann
weights corresponds to the `root-of-unity' cases. In the specific case of dimer models
on planar bipartite graphs, the relation between
Arctic curves and algebraic curves was investigated in
\cite{KOS-06,KO-06,KO-07}.

The outline of the paper is as follows. In the next section we recall the result
for the Arctic curve of the domain-wall six-vertex model in its disordered
regime. The specificities of the `root-of-unity' cases are explained in section 3.
Some interesting particular cases are exposed in section 4. We end
up by discussing the upper limit bound for the degree of the algebraic
equation for the Arctic curve for generic `root-of-unity' case in section 5.
Finally, some technical aspects are discussed in the appendix.

\section{Arctic curve of the disordered regime}

We start with recalling the result of paper \cite{CP-09} on the arctic curve of the
six-vertex model with domain wall boundary conditions in its disordered regime. We
first recall how the Arctic curve arises in the model and next explain the
result.

We consider the six-vertex model on a square lattice formed by intersection of
equal number of vertical and horizontal lines, with all states on the
boundaries fixed in the special way, called the domain wall boundary
conditions \cite{K-82}. Using the arrow language of states on edges, these boundary
conditions mean that all arrows on top and bottom boundaries are incoming,
while on  the left and right boundaries they are outgoing.

The Boltzmann weights, usually denoted as $a$, $b$ and $c$, are parameterized
as
\begin{equation}\label{abc}
a=\sin(\lambda+\eta),\qquad b=\sin(\lambda-\eta),\qquad c=\sin2\eta,
\end{equation}
where
$\lambda$ and $\eta$ are the `rapidity' variable and the `crossing'
parameter, respectively. The disordered regime
corresponds to $\lambda$
and $\eta$ real and satisfying
\begin{equation}
\eta<\lambda<\pi-\eta,\qquad
0<\eta<\pi/2.
\end{equation}
It is useful to recall that physical regimes of the six-vertex model
are characterized by the parameter $\Delta=(a^2+b^2-c^2)/2ab$,
and the disordered regime corresponds to $|\Delta|<1$; in our parametrization we
have $\Delta=\cos2\eta$.

The Arctic curve describes spatial separation of phases, of ferroelectric order
and disorder. The effect of separation of phases is
related to the fact that in the domain-wall six-vertex model
ordered configurations on the boundary can induce,
through the ice-rule, a macroscopic order inside the lattice.

The notion of phase separation, and hence the Arctic curve, acquires a precise
meaning in the scaling limit, i.e.,  when
the number of lines of the lattice (in each direction) tends to infinity and the
lattice spacing vanishes, while the total size of the system (in each direction)
is kept fixed. For the domain-wall six-vertex model
one can assume that the lattice is scaled to the square $[0,1]\times[0,1]$.
To describe points of this square we use coordinates $(x,y)$ with $x,y\in [0,1]$.
As in \cite{CP-09}, we consider the coordinate system in which the $y$-axis is reversed,
with the origin at the top left corner of the square.

The phase separation of domain-wall six-vertex model in its
disordered regime is characterized by appearance of
five regions in the scaling limit:
four regions of ferroelectric order, $\FO_\NW$, $\FO_\NE$, $\FO_\SE$, and $\FO_\SW$
in the four corners of the square, and one region of disorder $\DI$,
in the centre. The region of disorder is sharply delimited by a
curve  $\arctic$, called the Arctic curve. The Arctic curve and the square
have four contact points, located each one on a side of the square.
The arctic curve in the disordered regime consists of four portions,
\begin{equation}
\arctic=\varGamma_\NW\cup\varGamma_\NE\cup\varGamma_\SE\cup\varGamma_\SW,
\end{equation}
where $\varGamma_i$ separates the region $\FO_i$  ($i=\NW,\NE,\SE,\SW$) from
the internal region of disorder $\DI$.
Due to the symmetries of the model (e.g., crossing symmetry) given,
for example, portion $\varGamma_\NW$,
one can easily obtain the remaining three portions
$\varGamma_\NE$, $\varGamma_\SE$, and $\varGamma_\SW$
(for details, see \cite{CP-09}, section 2) and hence obtain the whole
curve $\arctic$.

To describe the curve $\varGamma_\NW$, one can introduce a
function $\Upsilon(x,y;\lambda)$ where $x$ and $y$ are coordinates in the
scaling limit and $\lambda$ is the parameter of the weights,
\begin{equation}\label{GNW}
\varGamma_\NW:\ \Upsilon(x,y;\lambda)=0,\qquad x,y\in[0,\kappa].
\end{equation}
Due to symmetries of the model, this function obeys
$\Upsilon(x,y;\lambda)=\Upsilon(y,x;\lambda)$.
The quantity $\kappa=\kappa(\lambda)$ gives
location of the contact points of the Arctic curve, in particularly, the
points $(\kappa,0)$ and $(0,\kappa)$ are the end-points of the curve
$\varGamma_\NW$. Explicitly (see \cite{CP-09}, section 3), it reads
\begin{equation}
\kappa=\frac{\alpha\cot\alpha(\lambda-\eta)-\cot(\lambda+\eta)}
{\cot(\lambda-\eta)-\cot(\lambda+\eta)},
\end{equation}
where
\begin{equation}
\alpha = \frac{\pi}{\pi-2\eta}.
\end{equation}
In the case of $\lambda=\pi/2$, i.e., when the weights $a$ and
$b$ are equal, see \eqref{abc}, one has $\kappa=1/2$ for all values of the
parameter $\eta$.

The explicit form of the function $\Upsilon(x,y;\lambda)$ is
significantly determined by the value of $\eta$. Furthermore,
it turns out that function $\Upsilon(x,y;\lambda)$ for generic values
of $\eta$ is a non-algebraic, or transcendental, function. This follows
from the parametric solution for the curve $\varGamma_\NW$, obtained in \cite{CP-09}.

Namely, let $\zeta$ be real and taking values in the interval $[0,\pi-\lambda-\eta]$.
Let us consider the function $f(x,y;\lambda;\zeta)$ which depends on $x$ and $y$
linearly, and is given explicitly by the formula
\begin{multline}\label{ef}
f(x,y;\lambda;\zeta)= \frac{x \sin 2 \eta }{\sin(\zeta+\lambda-\eta)
\sin(\zeta+\lambda+\eta)}
+ \frac{y \sin 2 \eta }{\sin \zeta \sin(\zeta+2\eta)}
\\
- \frac{\sin(\lambda+\eta)}{\sin\zeta\sin(\zeta+\lambda+\eta)}
+ \frac{\alpha \sin\alpha(\lambda-\eta)}{\sin\alpha\zeta\sin\alpha(\zeta+\lambda-\eta)}.
\end{multline}
The parameter $\zeta$ parameterizes the curve $\varGamma_\NW$
as it runs over the interval $[0,\pi-\lambda-\eta]$ while the curve is given in
the parametric form
\begin{equation}
x=X(\zeta),\qquad y=Y(\zeta).
\end{equation}
Writing simply $f(\zeta)$ for
function $f(x,y;\lambda;\zeta)$, the functions $X(\zeta)$ and $Y(\zeta)$ correspond to
the solution, for unknowns $x$ and $y$, of the linear system of equations
\begin{equation}\label{ffprime}
f(\zeta)=0,\qquad  f'(\zeta)=0,
\end{equation}
where the prime denotes derivative. For later use, we mention that
equations \eqref{ffprime} are equivalent to
the condition that function $f(\zeta)$ must have a double root;
each point of the curve $\varGamma_\NW$ corresponds to
a real value of this double root, in the interval $[0,\pi-\lambda-\eta]$.

Functions $X(\zeta)$ and $Y(\zeta)$ are
related to each other as $X(\zeta)=Y(\pi-\lambda-\eta-\zeta)$ that follows from the
obvious property
\begin{equation}
f(x,y;\lambda;\zeta)=f(y,x;\lambda;\pi-\lambda-\eta-\zeta),
\end{equation}
reflecting the $x\leftrightarrow y$ symmetry of the curve $\varGamma_\NW$.
Explicit expressions for functions $X(\zeta)$ and $Y(\zeta)$
can be found in \cite{CP-09}, see equations (6.16)--(6.19) therein; for our
discussion below expression \eqref{ef} and equations \eqref{ffprime} are
sufficient.

\section{`Root-of-unity' cases}

Let us now focus on the case when the parameter $\eta$ is such that the parameter
$\alpha$ is a rational number
\begin{equation}\label{nd}
\alpha=\frac{n}{d}\qquad (d<n),
\end{equation}
where $n$ and $d$ are some co-prime integers.
The parameter $\eta$ reads
\begin{equation}\label{eta}
\eta=\frac{\pi}{2}\cdot\frac{n-d}{n},
\end{equation}
and it corresponds to the so-called
six-vertex model at a `root of unity' (the root of unity here
is the deformation parameter $q$ of the underlying quantum group,
$q=\exp2\rmi\eta=-\exp(-\rmi\pi/\alpha)$).

To make the subsequent discussion simpler,
we start with rewriting function \eqref{ef} in a more symmetric form.
Namely, we
introduce parameters $\phi$ and $\varkappa$ related
to parameters $\zeta$ and $\lambda$ by a linear change of variables
\begin{equation}
\phi=\zeta-\varkappa,\qquad \varkappa=\frac{\pi-\lambda-\eta}{2}.
\end{equation}
Introducing function $g(x,y;\varkappa;\varphi):=f(x,y;\lambda;\zeta)$,
which we shall write simply as $g(\phi)$, we arrive at a more symmetric expression
\begin{multline}\label{gphi}
g(\phi)= \frac{x\sin2\eta}{\sin(\varkappa-\phi)\sin(\varkappa+2\eta-\phi)}
+\frac{y\sin2\eta}{\sin(\varkappa+\phi)\sin(\varkappa+2\eta+\phi)}
\\
-\frac{\sin2\varkappa}{\sin(\varkappa+\phi)\sin(\varkappa-\phi)}
+\frac{\alpha\sin2\alpha\varkappa}{\sin\alpha(\varkappa+\phi)
\sin\alpha(\varkappa-\phi)}.
\end{multline}
Note that new parameters run over the values
$\phi\in[-\varkappa,\varkappa]$ and $\varkappa\in(0,\pi/2\alpha)$, and
$\alpha$ ($1<\alpha<\infty$) in \eqref{gphi}
is still arbitrary. Evidently, equations \eqref{ffprime}
now read $g(\phi)=0$ and $g'(\phi)=0$ and we have simply reformulated
the parametric expression for the curve $\varGamma_\NW$
just by shifting the parameter of the curve.

Using the well-known identity
\begin{equation}
\sin n\phi=2^{n-1}\prod_{j=0}^{n-1}\sin\left(\phi+\frac{\pi j}{n}\right),
\end{equation}
where $n$ is an arbitrary positive integer, one can easily derive the identity
\begin{equation}\label{sinnd}
\sin\frac{n}{d}(\varkappa\pm\phi)=
\left(\sin\frac{\phi}{d}\right)^n
\sin\frac{n}{d}\varkappa
\prod_{j=0}^{n-1}\left(\cot\frac{\phi}{d}
\pm\cot\left(\frac{\varkappa}{d}+\frac{\pi j}{n}\right)\right).
\end{equation}
Formally setting here $n=d$, we also have the identity
\begin{equation}\label{sindd}
\sin(\varkappa\pm\phi)=
\left(\sin\frac{\phi}{d}\right)^d\sin\varkappa
\prod_{k=0}^{d-1}\left(\cot\frac{\phi}{d}
\pm\cot\left(\frac{\varkappa}{d}+\frac{\pi j}{d}\right)\right).
\end{equation}
Let us now rewrite the function $g(\phi)$ for the `root-of-unity' cases
using relations \eqref{sinnd} and \eqref{sindd}. Denoting
\begin{equation}
t=\cot\frac{\phi}{d},
\end{equation}
and taking into account that $\sin^2(\phi/d)=(t^2+1)^{-1}$,
we arrive at the following formula:
\begin{multline}\label{tformlong}
g(\phi)=
x\rho\,\frac{\big(t^2+1\big)^d}{\prod_{k}^{}(t-v_k)(t-u_k)}
+y\rho\,\frac{\big(t^2+1\big)^d}{\prod_{k}^{}(t+v_k)(t+u_k)}
\\
-2\cot\varkappa\,
\frac{\big(t^2+1\big)^d}{\prod_{k}^{}(t-v_k)(t+v_k)}
+2\alpha\cot\alpha\varkappa\,
\frac{\big(t^2+1\big)^n}{\prod_{j}^{}(t-w_j)(t+w_j)}.
\end{multline}
Here
\begin{equation}
\rho=\frac{\sin2\eta}{\sin\varkappa\sin(\varkappa+2\eta)}
\end{equation}
and the numbers $v_k$, $u_k$ ($k=0,\dots,d-1$)
and $w_j$ ($j=0,\dots,n-1$) are
\begin{align}\label{vuw}
v_k&=\cot\left(\frac{\varkappa}{d}+\frac{\pi k}{d}\right),
\notag\\
u_k&=\cot\left(\frac{\varkappa}{d}+\frac{2\eta}{d}+\frac{\pi k}{d}\right),
\notag\\
w_j&=\cot\left(\frac{\varkappa}{d}+\frac{\pi j}{n}\right).
\end{align}
The products in \eqref{tformlong} are taken over the indicated values of the
integers $j$ or $k$.

Clearly, from formula \eqref{tformlong} it follows that function
$g(\phi)$ is a rational function of the variable $t$, with
the following structure
\begin{equation}\label{tformshort}
g(\phi)=\frac{\big(t^2+1\big)^d}{Q(t)}P(t).
\end{equation}
Here $P(t)$ and $Q(t)$ are polynomials:  $Q(t)$ is the common
denominator of the four terms in \eqref{tformlong} while $P(t)$ is the
resulting numerator.
Evidently, the condition that function  $g(\phi)$ has a double real root
in the interval $\phi\in[-\varkappa,\varkappa]$,
which provides a parametric form for the curve $\varGamma_\NW$,
now translates into the condition that polynomial $P(t)$ has
a double  real root in the interval $t\in[-\infty,-v_0]\cup[v_0,\infty]$.

It is well known a polynomial has a double root (not necessarily real)
if and only if its  discriminant is equal to zero. The discriminant, in turn,
is a homogenous polynomial in the coefficients of the polynomial. The
coefficients of $P(t)$ are linear functions of $x$ and $y$,
and therefore requiring the discriminant of $P(t)$ to be equal to zero
provides an equation for an algebraic curve. The portion of such algebraic curve
corresponding to the additional requirement that  the double root
must lie in the  real interval $t\in[-\infty,-v_0]\cup[v_0,\infty]$
is the portion $\varGamma_{\NW}$ of the Arctic curve,
lying in the region $x,y\in[0,\kappa]$ of the unit square.

To be more specific,  let  $P(t)$ be a polynomial of degree $m$ ($m\geq 2$) with
non-vanishing leading coefficient, $p_m\ne 0$.
The discriminant of $P(t)$  can be written as
\begin{equation}\label{DtPP'}
D_m(P)=(-1)^{m(m-1)/2}\det S_{m-1,m-1}(\tilde P,P').
\end{equation}
where $S_{m-1,m-1}(\tilde P,P')$ is the Sylvester matrix of the two  polynomials
$\tilde P(t)=mP(t)-tP'(t)$  and $P'(t)$,  where the prime denotes derivative.
Further details are given in the appendix.
The condition that $P(t)$ has a double root is just
\begin{equation}
D_m(P)=0.
\end{equation}
In view of the above considerations, we conclude that
in the  `root-of-unity' cases this equation contains equation \eqref{GNW},
with $D_m(P)$ providing  (modulo the problem of reducibility)
an explicit expression for function $\Upsilon(x,y;\lambda)$.

In the next section we consider several interesting
examples of the `root-of-unity' cases, namely,
we consider the cases of $\alpha=2,3,3/2,4,5/2$.
The case of $\alpha=2$ is the free-fermion point and the Arctic curve is the
Arctic Ellipse (or Arctic Circle, at $\lambda=\pi/2$), see \cite{CP-07a,CP-08}
and references therein.
The cases $\alpha=3,3/2$, at $\lambda=\pi/2$,
were already treated in papers \cite{CP-08,CP-09}; here we discuss
the case of generic $\lambda$. The cases of $\alpha=4, 5/2$ are considered for
the first time.

In particular, we find that the degree
of $P(t)$ is significantly determined by the integers $n$ and $d$
defining the parameter $\alpha=n/d$. We also find that
at $\lambda=\pi/2$ a simplification in  the degree may occur.
Motivated by the examples considered below,
we discuss the upper bounds on the degree of algebraic equations
for the Arctic curve in section 5.

\section{Examples}

\subsection{Case $\alpha=2$ ($\eta=\pi/4$ or $\Delta=0$)}

In this case the numbers $v_k$, $u_k$ ($k=0,\dots,d-1$) and $w_j$ ($j=0,\dots,n-1$),
see \eqref{vuw}, which describe location of poles are ($d=1$ and $n=2$):
\begin{equation}
v_0=\cot\varkappa,\qquad
u_0=-\tan\varkappa,\qquad
w_0=\cot\varkappa,\qquad
w_1=-\tan\varkappa.
\end{equation}
Since $v_0=w_0$ and $u_0=w_1$, the common denominator is
given by the denominator of the last term in \eqref{gphi}, and hence the polynomial
$P(t)$ is quadratic. Computing, we obtain
\begin{equation}
P(t)=\frac{2}{\sin2\varkappa}\big[(x+y-1+\cos2\varkappa)t^2
+2 \cot2\varkappa\cdot t-(x+y-1-\cos2\varkappa)\big].
\end{equation}
Evaluating the discriminant of $P(t)$ and cancelling an overall factor, for
the Arctic curve we find
\begin{equation}\label{AE}
\frac{(x+y-1)^2}{\cos^2 2\varkappa}+\frac{(x-y)^2}{\sin^2 2\varkappa}-1=0.
\end{equation}
This is of course the well known result of the Arctic Ellipse
(see, e.g., \cite{CP-07a} and references therein).

\subsection{Case $\alpha=3$ ($\eta=\pi/3$ or $\Delta=-1/2$)}

In this case the denominators in \eqref{tformlong} are governed by numbers
$v_0$, $u_0$, $w_0$, $w_1$, and $w_2$, and we have relations
\begin{equation}
v_0=w_0,\qquad u_0=w_2.
\end{equation}
Thus the common denominator is given by the denominator of the last term
in \eqref{tformlong}
and in this case $P(t)$ is a polynomial of degree $4$. After some algebra
we find
\begin{align}
P(t)&=\frac{2}{\sin\!^2 3\nu}\big\{
\sin3\nu\big[\sin\nu\cdot\Lambda  -\cos\nu (1-4\cos2\nu)\big]t^4
\notag\\ &\quad
+4 \sin\!^2\nu\sin2\nu\cdot\Theta\, t^3
-\big[3\cos2\nu\cdot\Lambda +\sin2\nu(1-4\cos4\nu)\big]t^2
\notag\\ &\quad
-4 \cos\!^2\nu\sin2\nu\cdot\Theta\, t
+\cos3\nu\big[\cos\nu\cdot \Lambda +\sin\nu (1+4\cos2\nu)\big]
\big\}
\end{align}
where
\begin{equation}
\Lambda=\sqrt{3}(1-x-y),\qquad
\Theta=\sqrt{3}(x-y),\qquad
\nu=\varkappa+\frac{\pi}{3}.
\end{equation}
Evaluating the discriminant of this polynomial and equating it to zero gives
the equation for Arctic curve. In this case it is of degree $6$. The explicit
expression for generic $\lambda$ (or $\varkappa$) is too cumbersome to be
given here (though it is just of degree $6$, each coefficient is a
complicated function of $\varkappa$).

In the special case of $\lambda=\pi/2$, that corresponds to the choice
$\varkappa=\pi/12$,
the expression for the Arctic curve simplifies considerably and one arrives at equation
\begin{multline}
324 \big(x^6 + y^6\big)+ 1620 \big(x^5 y + x y^5\big)
+ 3429 \big(x^4 y^2 + x^2 y^4\big) + 4254 x^3 y^3
\\
-972\big(x^5+y^5\big)
-1458\big(x^4y+xy^4\big)-2970\big(x^3y^2+x^2y^3\big)
-6147\big(x^4+ y^4\big)
\\
-9150\big(x^3 y + x y^3\big)
-17462 x^2 y^2 + 13914 \big(x^3+ y^3\big)
+ 24086 \big(x^2 y+x y^2\big)
\\
-11511\big(x^2+ y^2\big)-17258xy+4392(x+y)-648
=0.
\end{multline}
This formula describes the limit shape of $3$-enumerated alternating-sign
matrices and it was obtained for the first time in \cite{CP-08}.

\subsection{Case $\alpha=3/2$ ($\eta=\pi/6$ or $\Delta=1/2$)}

In this case, with  $\lambda$ (or $\varkappa$) is generic, the polynomial
$P(t)$ is of degree $6$. Indeed, for $d=2$ and $n=3$ the
denominators in \eqref{tformlong} are controlled by
seven numbers $v_0$, $v_1$, $u_0$, $u_1$, $w_0$, $w_1$, and $w_2$ among those
there are two pairs of coinciding ones,
\begin{equation}\label{}
v_0=w_0, \qquad u_1=w_2.
\end{equation}
Already the resulting polynomial $P(t)$ is too lengthy in its coefficients to
be presented here, not even to mention its discriminant, giving the
Arctic curve, which is of degree $10$ in this case.

This curve is however interesting in the fact that
it factorizes when one specializes to the case of $\lambda=\pi/2$. Namely,
an expression of degree $10$ factorizes into several factors
among which  only one factor, of degree $2$, corresponds to the Arctic curve,
while the remaining factors are non-vanishing for the allowed values
of $x$ and $y$.

The origin of such factorization becomes evident if one comes back to formula
\eqref{tformlong}. Indeed, at $\lambda=\pi/2$, which corresponds to $\varkappa=\pi/6$,
we have two more additional relations
\begin{equation}
v_1=-w_1,\qquad u_0=-w_2.
\end{equation}
Correspondingly, in this special case the polynomial $P(t)$ is quadratic, and
it reads
\begin{equation}
P(t)=\sqrt{3}(x+y-2+\sqrt{3})t^2 +6 (x-y)t -\sqrt{3}(x+y-2-\sqrt{3}).
\end{equation}
Evaluating the discriminant, and setting  it equal to zero, we obtain equation
\begin{equation}\label{ASMs}
(x+y-2)^2+3(x-y)^2-3=0.
\end{equation}
Evidently, this equation describes an ellipse inscribed into square $[0,2]\times[0,2]$;
we recall that the Arctic curve is given only by a portion of this ellipse, for
which $x,y\in[0,1/2]$.

Equation \eqref{ASMs} is an important example of  Arctic curve, since it describes
the limit shape of large alternating-sign matrices (within $1$-enume\-ration scheme).
Equation \eqref{ASMs} was obtained previously in papers \cite{CP-08,CP-09}.
Interestingly, expanding terms in \eqref{ASMs} we obtain
\begin{equation}
4x^2-4xy+ 4y^2-4(x+y) +1=0,
\end{equation}
which differs just in a single term, namely $-4xy$, from equation \eqref{AE} at
$\lambda=\pi/2$ (corresponding to $\varkappa=\pi/8$ in that case).

\subsection{Case $\alpha=4$ ($\eta=3\pi/8$ or  $\Delta=-\sqrt{2}/2$)}

In this case the denominators in \eqref{tformlong} are controlled by the
numbers $v_0$, $u_0$ and  $w_0,\dots,w_3$. Similarly to the cases of
$\alpha=2$ and $\alpha=3$, the common denominator is
given by denominator of the last term in \eqref{tformlong}, due to the
relations $v_0=w_0$ and $u_0=w_3$. Now the polynomial $P(t)$ has
degree $6$, and the  Arctic curve is given by equation of degree $10$.
At $\lambda=\pi/2$ no simplification of the degree of $P(t)$ is
possible, and the Arctic curve remains of degree $10$; the explicit
expression for the equation, which  contains irrational coefficients involving
$\sqrt{2}$,  is lengthy and uninformative.

The case of $\alpha=4$ has been included in the list of our examples to show some
common properties of the cases of integer values of $\alpha$, namely, that the
common denominator is given by the denominator of the last term in
\eqref{tformlong}, and that no simplification of the degree occurs at
$\lambda=\pi/2$. For these reasons considering the integer values of
$\alpha$  is in fact not illuminating; the cases
$\alpha=2$ and $\alpha=3$ are somewhat exceptional, in view of their
well-known relations to important  enumeration problems in combinatorics.

\subsection{Case $\alpha=5/2$ ($\eta=3\pi/10$ or $\Delta=(1-\sqrt{5})/4$)}

In this case, for generic values of $\lambda$, the polynomial $P(t)$ is of degree $10$
and, correspondingly, the Arctic curve is given by equation of degree $18$.
We restrict ourselves henceforth to the case of $\lambda=\pi/2$ where
the degree of polynomial $P(t)$ is just $6$. Indeed, for $\lambda=\pi/2$
we have in this case $\varkappa=\pi/10$ and for the numbers
$v_0$, $v_1$, $u_0$, $u_1$ and $w_0,\dots,w_4$ determining the denominators of terms in
\eqref{tformlong} we have relations
\begin{equation}
v_0=w_0,\qquad
u_1=w_4,\qquad
v_1=-w_2,\qquad
u_0=-w_3.
\end{equation}
Due to these relations the common denominator $Q(t)$ in this case is
given by denominator of the last term in \eqref{tformlong}. Computing the
polynomial $P(t)$ we obtain the expression
\begin{align}\label{Palpha52}
P(t)&=\rho\bigg\{
\bigg(x+y-\frac{1+\sqrt{5}}{2}+\sqrt{5} \sigma\bigg)t^6
+8  \sigma  (x-y)t^5
\notag\\ &\qquad
- \bigg(\big(3+4 \sqrt{5}\big)(x+y)+\frac{17-15 \sqrt{5}}{2} -3\sqrt{5} \sigma \bigg) t^4
-16  \sigma  (x-y) t^3
\notag\\ &\qquad
+\bigg(\big(3+4 \sqrt{5}\big)(x+y)+\frac{17-15 \sqrt{5}}{2}+3\sqrt{5} \sigma \bigg) t^2
+8  \sigma  (x-y)t
\notag\\ &\qquad
-\bigg(x+y -\frac{1+\sqrt{5}}{2} -\sqrt{5} \sigma\bigg)
\bigg\},
\end{align}
where $\rho=\sqrt{2(5+\sqrt{5})}$ and $\sigma=\sqrt{\frac{5-\sqrt{5}}{8}}=\sqrt{5}\rho^{-1}$.

The discriminant of polynomial \eqref{Palpha52}, which is of degree $10$ in $x$ and $y$,
appears to  factorize into two polynomials, of degrees $2$ and $8$. In other words, the
corresponding algebraic curve has two components. The quadratic factor is
\begin{multline}
2 x^2- (1 - \sqrt{5}) x y + 2 y^2 - 2x-2y +3 - \sqrt{5}
\\
=\frac{3 + \sqrt{5}}{4} \big(x + y - 3 + \sqrt{5}\big)^2 + \frac{5 -\sqrt{5}}{4}  (x-y)^2
\end{multline}
where the second line shows that it is positive for all values of $x$ and
$y$, except point $x=y=(3-\sqrt{5})/2$ where it vanishes. It is easily verified
that this component (consisting in just one point) corresponds to the
double root $t=i$ (or $t=-i$) of $P(t)$. Hence, the Arctic curve
is given by the other component, i.e. by the factor degree of $8$. The equation for the
Arctic curve reads:
\begin{multline}
128 \big(x^8+y^8\big)
-512 \big(1-\sqrt{5}\big)\big(x^7 y+x y^7\big)
+1536 \big(x^6 y^2+x^2 y^6\big)
\\
+512 \big(25-9 \sqrt{5}\big) \big(x^5 y^3+x^3 y^5\big)
+256 \big(83-32 \sqrt{5}\big) x^4 y^4
-512 \big(x^7+y^7\big)
\\
+512 \big(13-8 \sqrt{5}\big) \big(x^6 y+x y^6\big)
-512 \big(65-24 \sqrt{5}\big) \big(x^5 y^2+x^2 y^5\big)
\\
-512 \big(267-112 \sqrt{5}\big) \big(x^4 y^3+x^3 y^4\big)
-96 \big(74-29 \sqrt{5}\big)\big(x^6+y^6\big)
\\
+64\big(435-182 \sqrt{5}\big) \big(x^5 y+x y^5\big)
-32 \big(2874-1313 \sqrt{5}\big) \big(x^4 y^2+x^2 y^4\big)
\\
+128 \big(7233-3214 \sqrt{5}\big) x^3 y^3
+32 \big(722-261 \sqrt{5}\big) \big(x^5+y^5\big)
\\
+32 \big(6132-2663 \sqrt{5}\big) \big(x^4 y+x y^4\big)
-64 \big(6911-3154 \sqrt{5}\big) \big(x^3 y^2+x^2 y^3\big)
\\
-4 \big(12403-5205 \sqrt{5}\big) \big(x^4+y^4\big)
-4 \big(97677-42727 \sqrt{5}\big)\big(x^3 y+x y^3\big)
\\
-16 \big(18381-8009 \sqrt{5}\big) x^2 y^2
+360 \big(167-77 \sqrt{5}\big) \big(x^3+y^3\big)
\\
+16\big(58205-26052 \sqrt{5}\big) \big(x^2 y+x y^2\big)
-\big(31813-15751 \sqrt{5}\big) \big(x^2+y^2\big)
\\
-2 \big(385088-173179 \sqrt{5}\big) x y
+\big(5689-3283 \sqrt{5}\big) (x+y)
+85 \sqrt{5}+87
=0.
\end{multline}

Even though the resulting equation is not illuminating, the whole example
of the case of $\alpha=5/2$ appears very instructive. Indeed, we observe again,
like in the case of $\alpha=3/2$,
a reduction in the degree of polynomial $P(t)$ at $\lambda=\pi/2$.
Furthermore, we observe that the Arctic curve can be of a lower degree, with
respect to the degree expected from the one of $P(t)$, due to additional factorization.

As we show below, the drop down of degrees of polynomials $P(t)$ is a
common feature of all cases with half-integer values of $\alpha$ at
$\lambda=\pi/2$. On the other hand the additional factorization of its discriminant
cannot be tackled by simple means. It corresponds
to the existence of a family of complex double roots for $P(t)$,
and hence to the existence of an additional component in the algebraic curve
defined by the condition of vanishing resultant for $P(t)$.
The identification and elimination of such additional component can be performed,
at least in principle, in each specific case, but no general procedure is available.

\section{Upper bounds for degrees of the curves}

In the examples considered above it is evident that the degree of polynomial
$P(t)$ and hence the degree for the Arctic curve as equation on $x$ and $y$,
essentially depends on the values of $d$ and $n$, the co-prime integers determining
value of the parameter $\alpha$.

Using formulae of section 3 it is fairly easy to find an expression for the
common denominator in \eqref{tformshort}. In this way one can infer general properties
of the numerator and hence find the degree of $P(t)$.

Assuming that we are in the case of generic
$\lambda$ (or $\varkappa$) we have relations
\begin{equation}\label{v0w0}
v_0=w_0,\qquad u_{d-1}=w_{n-1}.
\end{equation}
Indeed, the first relation is immediate, while the second follows from $\pi$-periodicity
of cotangent function and formula
\begin{equation}
\frac{2\eta}{d}=\frac{\pi}{d}-\frac{\pi}{n}.
\end{equation}
Since $\varkappa$ is assumed to be arbitrary, no further relations
among numbers $v_0,\dots,v_{d-1}$, $u_0,\dots,u_{d-1}$, and $w_0,\dots,w_{n-1}$ are
possible.

Having relations \eqref{v0w0} we obtain that the common denominator
$Q(t)$ in formula \eqref{tformshort} reads
\begin{equation}\label{Qgen}
Q(t)=\prod_{j=0}^{n-1}(t-w_j)(t+w_j)\prod_{k=1}^{d-1}(t-v_k)(t+v_k)
\prod_{k=0}^{d-2}(t-u_k)(t+u_k).
\end{equation}
Using this expression one can easily see that each of the four terms in
\eqref{tformlong} contributes to the numerator with the terms of the same
degree in $t$. The result can be summarised as follows.
\begin{proposition}
In the case of $\alpha=n/d$ where $n$ and $d$ are co-prime integers and
$\lambda$ is generic, the polynomial P(t) in \eqref{tformshort}
has degree less than  or equal to $2d+2n-4$.
\end{proposition}
\begin{proof}
A comparison of formulae \eqref{tformlong}, \eqref{tformshort}, and
\eqref{Qgen} shows that each of the four terms in expression
\eqref{tformlong} contributes to the numerator $P(t)$ a term of degree
$2d+2n-4$. For some values of $\varkappa$ the polynomials $P(t)$ and $Q(t)$
may have common roots with respect to the variable $t$, hence the number
$2d+2n-4$ provides the upper bound for the degree of $P(t)$.
\end{proof}

As a consequence we have the following upper
bound on the degree of the Arctic curve in the generic `root-of-unity' case.
\begin{corollary}
In the case of $\alpha=n/d$ where $n$ and $d$ are co-prime integers and
$\lambda$ is generic, the function $\Upsilon(x,y;\lambda)$ in \eqref{GNW} is
a polynomial in $x$ and $y$ of total degree less than or equal to $4(n+d)-10$.
\end{corollary}

Let us now address a more refined question about the degree of $P(t)$ in
the special case of $d=2$ and $n$ odd (i.e., when $\alpha$ is
half-integer) at $\lambda=\pi/2$. As we have seen in our examples above,
the degree drops down, in comparison with the case of
generic $\lambda$. We also observed that no such drop down of the degree
occurs, e.g., in the cases of integer values of $\alpha$.

Here we have the following result.

\begin{proposition}
If $\alpha=n/2$ where $n$ is odd, and $\lambda=\pi/2$, then
\begin{equation}\label{2n4}
\deg P(t)= 2n-4.
\end{equation}
\end{proposition}
\begin{proof}
We start with noting that for arbitrary $\alpha=d/n$
the value $\lambda=\pi/2$ corresponds to $\varkappa=\pi/4\alpha$ and hence we have
\begin{equation}\label{pi4n}
\frac{\varkappa}{d}=\frac{\pi}{4n}.
\end{equation}
Let us now set $d=2$ and consider $n$ odd. In addition to the relations
$v_0=w_0$ and $u_1=w_{n-1}$ which are just relations \eqref{v0w0}, in the presently
considered case we have, in virtue of \eqref{pi4n}, the relations
\begin{equation}
v_1=-w_{\frac{n-1}{2}},\qquad u_0=-w_{\frac{n+1}{2}}.
\end{equation}
Due to these two additional relations we see that the common denominator
$Q(t)$ in \eqref{tformshort} coincides with the denominator of  the last
term in \eqref{tformlong}, that is
\begin{equation}
Q(t)=\prod_{j=0}^{n-1}(t-w_j)(t+w_j).
\end{equation}
Recalling that $d=2$ and inspecting the degrees of the numerators we thus
arrive at \eqref{2n4}.
\end{proof}

Therefore we have obtained the following statement.
\begin{corollary}
If $\alpha=n/2$ where $n$ is odd, and $\lambda=\pi/2$, then
the Arctic curve is given by an equation in $x$ and $y$ of
degree less than or equal to $4n-10$.
\end{corollary}
Note that here   we cannot claim exact equality, conversely to formula
\eqref{2n4}.
As a matter of fact, the exact  degree of the curve
can hardly be controlled, since it is difficult to find in
general the number of components of the algebraic curve resulting from the
condition of vanishing discriminant of $P(t)$. Indeed, as we have learned from the
example of  $\alpha=5/2$, the occurrence of more than one component
immediately implies that the  Arctic curve has lower degree with respect to
the one implied by the degree of $P(t)$. An interesting open question
concerns  therefore the possibility of providing some classification for the
irreducible algebraic curves associated to the expression
$\Upsilon(x,y;\lambda)$ for the Arctic curve.

\section*{Acknowledgments}

The first author (FC) acknowledges partial support from MIUR, PRIN grant 2007JHLPEZ,
and from the European Science Foundation program INSTANS.
The third author (AGP) was supported by the Alexander von Humboldt Foundation.
AGP also acknowledges partial support
from INFN, Sezione di Firenze, from the Russian Foundation for Basic Research
(grant 10-01-00600), and from the Russian Academy of Sciences
program ``Mathematical Methods in Nonlinear Dynamics''.

\appendix

\section{Discriminant of a polynomial}

Let  $P(t)$ be a polynomial of degree $m$ ($m\geq 2$) with
non-vanishing leading coefficient, $p_m\ne 0$, and let $r_i$ ($i=1,\dots,m$)
be the roots of this polynomial,
\begin{equation}
P(t)= \sum_{k=0}^{m} p_k\, t^{k}=p_m \prod_{i=1}^m (t-r_i).
\end{equation}
The discriminant of polynomial $P(t)$, denoted as $D_m(P)$, where
the subscript refers to the degree of the polynomial, is defined as
(sometimes another  definition is used,
differing by an overall factor $(-1)^{m(m-1)/2}$  cf.~\cite{CLS-07,GKZ-94}):
\begin{equation}\label{DPdef}
D_m(P):= p_m^{2m-2}\prod_{i<k} (r_i-r_k)^2.
\end{equation}
It is clear that the condition that $P(t)$ has a double root is equivalent to:
\begin{equation}\label{DP0}
D_m(P)=0.
\end{equation}

We now  explain  how an  explicit expression can be obtained for $D_m(P)$,
from the knowledge of the coefficients (rather than of the roots) of polynomial $P(t)$.
For this aim it is useful to first
rephrase equation \eqref{DP0}  as the condition that two polynomials, say $P(t)$
and $P'(t)$ (see below for a slightly different choice) have a common root.

It is well known (see, e.g., \cite{CLS-07,GKZ-94}) that two arbitrary
polynomials have a common (in general, complex) root if and only if the
determinant of the Sylvester matrix, i.e. the resultant, of these
polynomials  vanishes.
To be more specific, let $A(t)$ and $B(t)$ be polynomials in $t$, of degrees
$m_A$ and $m_B$ respectively ($a_{m_A}\ne 0$ and $b_{m_B}\ne 0$),
\begin{equation}
A(t)= \sum_{k=0}^{m_A}a_k t^k,\qquad
B(t)=\sum_{k=0}^{m_B}b_k t^k.
\end{equation}
The Sylvester matrix associated with the polynomials $A(t)$ and $B(t)$ is an
$(m_A+m_B)$-by-$(m_A+m_B)$ matrix, denoted $S_{m_A,m_B}(A,B)$, where
the subscripts refers to the degrees of $A(t)$ and $B(t)$, which is given
by the formula:
\begin{equation}\label{SAB}
S_{m_A,m_B}(A,B)= \begin{pmatrix}
a_0 & a_{1} & \cdots & a_{m_A{-}1} & a_{m_A} & 0  &\cdots & 0\\
0 & a_0 & a_1 & \cdots & a_{m_A{-}1}  &a_{m_A} & \cdots & 0\\
\vdots & \ddots & \ddots & \ddots & \ddots & \ddots &\ddots & \vdots\\
0 & \cdots & 0& a_0 & a_1 & \cdots & a_{m_A{-}1}  & a_{m_A}\\
b_0 & b_1 & \dots & b_{m_B{-}1} &b_{m_B} & 0 & \cdots & 0\\
0 & b_0 & b_1 & \dots & b_{m_B{-}1} & b_{m_B} & \cdots & 0\\
\vdots & \ddots & \ddots & \ddots & \ddots & \ddots & \ddots &\vdots \\
0 & \cdots & 0 & b_0 & b_1 & \dots & b_{m_B{-}1} & b_{m_B}\\
\end{pmatrix}.
\end{equation}
The requirement that polynomials $A(t)$ and $B(t)$ have a common root is
equivalent to the condition
\begin{equation}\label{detS}
\det S_{m_A,m_B}(A,B)=0.
\end{equation}
The complete proof of this statement can be found  e.g. in  \cite{CLS-07,GKZ-94}.
Here we shall just show that  condition \eqref{detS} is necessary.
Indeed, let $\vec \tau_m$ denote
an $m$-component  column vector
with  entries $(\vec \tau_m)_i=t^{i-1}$ ($i=1,\dots,m$) then
\begin{equation}
S_{m_A,m_B}(A,B)\,\vec \tau_{m_A+m_B}  =
\begin{pmatrix}
A(t)\,\vec \tau_{m_B} \\[6pt]
B(t)\, \vec \tau_{m_A}\end{pmatrix}
= 0.
\end{equation}
Thus, if $A(t)$ and $B(t)$ simultaneously  vanish for some value of $t$,
then the Sylvester matrix has at least one vanishing eigenvalue and hence formula
\eqref{detS} follows.

Given polynomial $P(t)$, the standard way of obtaining its discriminant (which is
sometimes used as a definition of $D_m(P)$, instead of \eqref{DPdef}) is:
\begin{equation}\label{DPP'}
D_m(P)=\frac{(-1)^{m(m-1)/2}}{p_m}\det S_{m,m-1}(P,P').
\end{equation}
Taking into account that
\begin{equation}
P'(t)=mp_m t^{m-1}+ (m-1)p_{m-1}t^{m-2}+\cdots+ 2 p_2 t +p_1,
\end{equation}
let us consider polynomial $\tilde P(t)=mP(t)-tP'(t)$, which reads:
\begin{equation}
\tilde P(t)=p_{m-1} t^{m-1}+ 2 p_{m-2}t^{m-2}+\cdots+ (m-1)p_1 t + m p_0.
\end{equation}
Noting that the system of equations $P(t)=0$, $P'(t)=0$ can be  replaced by
the system of equations $\tilde P(t)=0$, $P'(t)=0$, we can therefore obtain
equation \eqref{DP0} using the Sylvester determinant associated with the
polynomials $\tilde P(t)$ and $P'(t)$. Furthermore, this can also be extended
to the expression for $D_m(P)$ regardless  of the requirement of its
vanishing. Indeed, using the standard properties of determinants, formula
\eqref{DPP'} can be simplified to
\begin{equation}
D_m(P)=(-1)^{m(m-1)/2}\det S_{m-1,m-1}(\tilde P,P'),
\end{equation}
which is exactly formula  \eqref{DtPP'}.

\bibliographystyle{amsplain}
\bibliography{vanni_bib}

\begin{bibdiv}
\begin{biblist}

\bib{CEP-96}{article}{
      author={Cohn, H.},
      author={Elkies, N.},
      author={Propp, J.},
       title={Local statistics for random domino tilings of the {A}ztec
  diamond},
        date={1996},
     journal={Duke Math. J.},
      volume={85},
       pages={117\ndash 166},
      eprint={math/0008243},
}

\bib{JPS-98}{article}{
      author={Jockush, W.},
      author={Propp, J.},
      author={Shor, P.},
       title={Random domino tilings and the arctic circle theorem},
      eprint={math/9801068},
}

\bib{CLP-98}{article}{
      author={Cohn, H.},
      author={Larsen, M.},
      author={Propp, J.},
       title={The shape of a typical boxed plane partition},
        date={1998},
     journal={New York J. Math.},
      volume={4},
       pages={137\ndash 165},
      eprint={math/9801059},
}

\bib{BGR-10}{article}{
      author={Borodin, A.},
      author={Gorin, V.},
      author={Rains, E.~M.},
       title={$q$-{D}istributions on boxed plane partitions},
        date={2010},
     journal={Selecta Math. (N.S.)},
      volume={16},
      number={4},
       pages={731\ndash 789},
      eprint={0905.0679},
         url={\doi{10.1007/s00029-010-0034-y}},
}

\bib{K-82}{article}{
      author={Korepin, V.~E.},
       title={Calculations of norms of {B}ethe wave functions},
        date={1982},
     journal={Commun. Math. Phys.},
      volume={86},
       pages={391\ndash 418},
}

\bib{I-87}{article}{
      author={Izergin, A.~G.},
       title={Partition function of the six-vertex model in the finite volume},
        date={1987},
     journal={Sov. Phys. Dokl.},
      volume={32},
       pages={878\ndash 879},
}

\bib{ICK-92}{article}{
      author={Izergin, A.~G.},
      author={Coker, D.~A.},
      author={Korepin, V.~E.},
       title={Determinant formula for the six-vertex model},
        date={1992},
     journal={J. Phys. A},
      volume={25},
       pages={4315\ndash 4334},
}

\bib{E-99}{article}{
      author={Eloranta, K.},
       title={Diamond ice},
        date={1999},
     journal={J. Stat. Phys.},
      volume={96},
       pages={1091\ndash 1109},
}

\bib{KZj-00}{article}{
      author={Korepin, V.~E.},
      author={Zinn-Justin, P.},
       title={Thermodynamic limit of the six-vertex model with domain wall
  boundary conditions},
        date={2000},
     journal={J. Phys. A},
      volume={33},
       pages={7053\ndash 7066},
      eprint={cond-mat/0004250},
}

\bib{Zj-00}{article}{
      author={Zinn-Justin, P.},
       title={Six-vertex model with domain wall boundary conditions and
  one-matrix model},
        date={2000},
     journal={Phys. Rev. E},
      volume={62},
       pages={3411\ndash 3418},
      eprint={math-ph/0005008},
}

\bib{Zj-02}{article}{
      author={Zinn-Justin, P.},
       title={The influence of boundary conditions in the six-vertex model},
      eprint={cond-mat/0205192},
}

\bib{SZ-04}{article}{
      author={Sylju{\aa}sen, O.~F.},
      author={Zvonarev, M.~B.},
       title={Monte-{C}arlo simulations of vertex models},
        date={2004},
     journal={Phys. Rev. E},
      volume={70},
       pages={016118},
      eprint={cond-mat/0401491},
}

\bib{AR-05}{article}{
      author={Allison, D.},
      author={Reshetikhin, N.},
       title={Numerical study of the $6$-vertex model with domain wall boundary
  conditions},
        date={2005},
     journal={Ann. Inst. Fourier (Grenoble)},
      volume={55},
       pages={1847\ndash 1869},
      eprint={cond-mat/0502314},
}

\bib{R-10}{article}{
      author={Reshetikhin, N.},
       title={Lectures on the integrability of the $6$-vertex model},
      eprint={1010.5031},
}

\bib{CP-08}{article}{
      author={Colomo, F.},
      author={Pronko, A.~G.},
       title={The limit shape of large alternating-sign matrices},
        date={2010},
     journal={SIAM J. Discrete Math.},
      volume={24},
      number={4},
       pages={1558\ndash 1571},
      eprint={0803.2697},
         url={\doi{10.1137/080730639}},
}

\bib{CP-09}{article}{
      author={Colomo, F.},
      author={Pronko, A.~G.},
       title={The arctic curve of the domain-wall six-vertex model},
        date={2010},
     journal={J. Stat. Phys.},
      volume={138},
      number={4},
       pages={662\ndash 700},
      eprint={0907.1264},
}

\bib{CPZj-10}{article}{
      author={Colomo, F.},
      author={Pronko, A.~G.},
      author={Zinn-Justin, P.},
       title={The arctic curve of the domain-wall six-vertex model in its
  anti-ferroelectric regime},
        date={2010},
     journal={J. Stat. Mech.},
       pages={L03002},
      eprint={1001.2189},
         url={\doi{10.1088/1742-5468/2010/03/L03002}},
}

\bib{KOS-06}{article}{
      author={Kenyon, R.},
      author={Okounkov, A.},
      author={Sheffield, S.},
       title={Dimers and amoebae},
        date={2006},
     journal={Ann. of Math.},
      volume={163},
       pages={1019\ndash 1056},
      eprint={math-ph/0311005},
}

\bib{KO-06}{article}{
      author={Kenyon, R.},
      author={Okounkov, A.},
       title={Planar dimers and {H}arnack curves},
        date={2006},
     journal={Duke Math. J.},
      volume={131},
       pages={499\ndash 523},
      eprint={math.AG/0311062},
}

\bib{KO-07}{article}{
      author={Kenyon, R.},
      author={Okounkov, A.},
       title={Limit shapes and the complex {B}urgers equation},
        date={2007},
     journal={Acta Math.},
      volume={199},
       pages={263\ndash 302},
      eprint={math-ph/0507007},
}

\bib{CP-07a}{article}{
      author={Colomo, F.},
      author={Pronko, A.~G.},
       title={The {A}rctic {C}ircle revisited},
        date={2008},
     journal={Contemp. Math.},
      volume={458},
       pages={361\ndash 376},
      eprint={0704.0362},
}

\bib{CLS-07}{book}{
      author={Cox, D.},
      author={Little, J.},
      author={O'Shea, D.},
       title={Ideals, varieties and algorithms},
     edition={3},
   publisher={Springer},
        date={2007},
}

\bib{GKZ-94}{book}{
      author={Gelfand, I.~M.},
      author={Kapranov, M.~M.},
      author={Zelevinsky, A.~V.},
       title={Discriminants, resultants and multidimensional determinants},
      series={Mathematics: Theory and Applications},
   publisher={Birkhauser},
        date={1994},
}

\end{biblist}
\end{bibdiv}

\end{document}